\documentclass[a4paper]{article}
\usepackage{eurosym}
\usepackage{amssymb}
\usepackage{amsmath}
\usepackage{amsfonts}
\usepackage{graphicx}
\usepackage{url}
\usepackage{geometry}
\usepackage{array}

\usepackage{color}
\usepackage{multirow}

\newlength{\Oldarrayrulewidth}
\newcommand{\Cline}[2]{%
  \noalign{\global\setlength{\Oldarrayrulewidth}{\arrayrulewidth}}%
  \noalign{\global\setlength{\arrayrulewidth}{#1}}\cline{#2}%
  \noalign{\global\setlength{\arrayrulewidth}{\Oldarrayrulewidth}}}

\numberwithin{equation}{section}

\newtheorem{theorem}{Theorem}
\numberwithin{theorem}{section}

\newtheorem{assumption}[theorem]{Assumption}

\newtheorem{corollary}[theorem]{Corollary}

\newtheorem{proposition}[theorem]{Proposition}

\newtheorem{remark}[theorem]{Remark}

\newenvironment{proof}[1][Proof]{\noindent\textbf{#1.} }{\ \rule{0.5em}{0.5em}}

\newdimen\dummy
\dummy=\oddsidemargin
\begin{document}

\title{Option Pricing in the Moderate Deviations Regime}

\author{Peter Friz \\ TU and WIAS Berlin \\ friz@math.tu-berlin.de
  \and Stefan Gerhold
   \\ TU Wien \\ sgerhold@fam.tuwien.ac.at
   \and Arpad Pinter\thanks{We gratefully acknowledge financial support from DFG grant FR2943/2 (P.~Friz)
resp.\ the
Austrian Science Fund (FWF) under grant P~24880 (S.~Gerhold, A.~Pinter). This work was first presented by one of us (Friz) at Global Derivates 2015 in Amsterdam, participants are thanked for their feedback.}
  \\ TU Wien \\ apinter@fam.tuwien.ac.at
}

\date{\today}

\maketitle

\begin{abstract}
We consider call option prices in diffusion models close to expiry, in an asymptotic regime (``moderately out of the money'') that interpolates between the well-studied cases of at-the-money options and out-of-the-money fixed-strike options. 
First and higher order small-time moderate deviation estimates of call prices
and implied volatility are obtained.
The expansions involve only simple expressions of the model parameters, and we show
in detail how to calculate them for generic local and stochastic volatility models.
Some numerical examples for the Heston model illustrate the accuracy of our results.
\end{abstract}


\section{Introduction}

Small-maturity approximations of option prices have been studied extensively in recent years. While out-of-the-money  calls with fixed strike and at-the-money calls have both been thoroughly investigated, there is a significant asymptotic regime lying between the two. 
It has received little attention, and, to the best of our knowledge, none at all in the classical diffusion case. The aim of the present paper is to fill this gap. The
``moderately out-of-the-money'' regime reflects the reality of quoted option prices
(strikes move closer to the money as expiry shrinks), while offering excellent analytic tractability.

To put our results into
perspective, we recall some well-known facts on option price approximations
close to expiry. We write $C=C(K,t)$ for the call price, and $c(k,t)$ if we wish to express it
as a function of log-moneyness~$k$: 
\begin{equation}\label{eq:not}
  C(S_0e^k,t)/S_0 = c(k,t).
\end{equation}
We start with the {\bf at-the-money} (short: {\bf ATM}) regime $k=0$. In the Black-Scholes model, writing $c(k,t) = c_{\,\mathrm{BS}}( 0,t;\sigma )$ with volatility parameter $\sigma>0$,
we have the following ATM call price behaviour,
\begin{equation*}
  c_{\,\mathrm{BS}}( 0,t;\sigma ) \sim \frac{\sigma \sqrt{t}}{\sqrt{2\pi }}, \quad
  t\downarrow 0.
\end{equation*}
The same is actually true~\cite{MuNu11} in a generic semimartingale model with
diffusive component (with spot volatility $\sigma_0=\sqrt{v_{0}}>0$):
\begin{equation}\label{eq:atm}
  c( 0,t) \sim \frac{\sigma_0\sqrt{t}}{\sqrt{2\pi }},\quad
t\downarrow 0,
\end{equation}%
and this translates to the generic ATM implied volatility formula (even in presence
of jumps, as long as $v_{0}>0$)
\begin{equation*}
\sigma_{\mathrm{imp}}^{2}( 0,t) =v_{0}+o( 1), \quad 
t\downarrow 0.
\end{equation*}%
Higher order terms in $t$ will be model dependent. For instance, in the Heston case, with variance dynamics $dV=-\kappa \left( V-%
\bar{v}\right) dt+\eta \sqrt{V}dW$,
implied volatility has the following ATM expansion:
\begin{align}
\sigma_{\mathrm{imp}}^{2}( 0,t)  &= v_{0}+a( 0) t+o(
t), \label{eq:iv Heston} \\
a( 0)  &=-\frac{\eta ^{2}}{12}\left( 1-\frac{\rho ^{2}}{4}%
\right) +\frac{v_{0}\rho \eta }{4}+\frac{\kappa }{2}\left( \bar{v}%
-v_{0}\right). \notag
\end{align}%
This is Corollary~4.4 in~\cite{FoJaLe12},
and we note that $a( 0) $ has no easy interpretation in terms of
the model parameters.

Relaxing $k=0$ to $k=o(\sqrt{t})$ amounts to what we dub 
``{\bf almost-ATM}''  (short: {\bf AATM}) regime.\footnote{The term``almost-ATM" seems new, but this regime was considered by a number of authors including \cite{CaCo14,MuNu11}.}
 (In particular, $k \approx t^\beta$ is in the AATM regime if and only if $\beta>1/2$.)
Again for generic semimartingale models
with diffusive component and spot volatility $\sigma_0>0$, it is easy to see
\cite{CaCo14,MuNu11} that the ATM asymptotics~\eqref{eq:atm}  imply the
following almost-ATM  asymptotics:
\begin{equation*}
  c( k_{t},t) \sim
   \frac{\sigma_0\sqrt{t}}{\sqrt{2\pi }}, \quad k_t = o(\sqrt{t}),\ t\downarrow0.
\end{equation*}%
This fails when $k_{t}$ ceases to be $o( \sqrt{t}) $. Indeed, for 
$k_{t}=\theta \sqrt{t}$ with constant factor $\theta>0$,
we have~\cite{CaCo14,MuNu11}:
\begin{equation*}
c( k_{t},t ) \sim \mathbb{E}[N( -\theta ,\sigma_0^{2}) ^{+}]\sqrt{t}.
\end{equation*}%
This, too, holds true in the stated semimartingale generality. In any case,
the proof is based on the L\'{e}vy case with non-zero diffusity $v_{0}$, and
the result follows from comparison results which imply that the difference
is negligible to first order. For a thorough discussion of the regime
$k=O(\sqrt{t})$ in the (local) diffusion case, see~\cite{PaPa16}. \\

Beyond this regime, call price asymptotics change considerably. For instance, take an additional
slowly diverging factor $\log ( 1/t) $,%
\begin{equation*}
k_{t}=\theta \sqrt{t\log \left( 1/t\right) }.
\end{equation*}%
Even in the Black-Scholes model, we now loose the $\sqrt{t}$-behaviour of call prices seen above and in fact
\begin{equation*}
c_{\,\mathrm{BS}}( 0,t;\sigma ) =t^{\frac{1}{2}+\frac{\theta ^{2}}{2\sigma
^{2}}}\ell( t),
\end{equation*}
for some slowly
varying function $\ell( t)$, see~\cite{MiTa12}.
On the other hand, in a genuine 
{\bf out-of-the-money} (short: {\bf OTM}) situation, with $k_{t}\equiv k>0$ fixed, option
values are exponentially small in diffusion models, and we are in the realm of {\it large deviation
theory}. For instance,
\begin{equation*}
c_{\,\mathrm{BS}}\left( k,t;\sigma \right) \approx \exp \left( - \frac{\Lambda_{\mathrm{BS}}(k)}{t} \right), \quad k>0\ \text{fixed},\ t\downarrow0,
\end{equation*}%
with $\Lambda_{\mathrm{BS}}( k) =\frac{1}{2}k^{2}/\sigma ^{2}$ in the Black-Scholes
model.\footnote{More precisely, $t\log c_{\,\mathrm{BS}}\left( k,t;\sigma \right) \sim -\Lambda_{\mathrm{BS}}(k), \quad k>0\ \text{fixed},\ t\downarrow0$.}
Similar results appear in the literature, with different level of
mathematical rigor, for other and/or generic diffusion
models~\cite{BeBuFl02,CaWu03,FoJa09,Pa15}. 

\begin{table}[]
\caption{Asymptotic behavior of short-maturity call options, $t \downarrow 0$}
\begin{tabular}{|l|l|l|l|l|}
\hline
\Cline{2pt}{4-4}
\textbf{Process Type}                                                                                                
& \textbf{\begin{tabular}[c]{@{}l@{}} ATM\\ {\footnotesize (at-the-money)} \\ \\ $K = S_0$ \end{tabular}} 
& \textbf{\begin{tabular}[c]{@{}l@{}} AATM {\footnotesize (almost}\\ {\footnotesize at-the-money)} \\ \\ $\log \tfrac{K}{S_0}  \approx (const)t^\beta $ \\ $\beta > 1/2$ \end{tabular}} 
& \multicolumn{1}{!{\vrule width 2pt}c!{\vrule width 2pt}}{\textbf{\begin{tabular}[c]{@{}l@{}} MOTM {\footnotesize (moderately}\\ {\footnotesize out-of-the-money)} \\ \\ $\log \tfrac{K}{S_0}  \approx (const)t^\beta $ \\ $0 < \beta < 1/2$ \end{tabular}}}
& \textbf{\begin{tabular}[c]{@{}l@{}} OTM\\ {\footnotesize (out-of-the-money)}\\ \\ $\log(K/S_0) \equiv k > 0$ \end{tabular}} \\ 

\hline
\textbf{Black-Scholes}                                                                                                & \begin{tabular}[c]{@{}l@{}} $O(\sqrt{t})$, \\ elementary\end{tabular}
& \begin{tabular}[c]{@{}l@{}} $O(\sqrt{t})$, \\ elementary\end{tabular}             & \multicolumn{1}{!{\vrule width 2pt}c!{\vrule width 2pt}}{$\exp\big(-\frac{const}{t^{1-2\beta}}\big)$}
& \begin{tabular}[c]{@{}l@{}} $\exp\big(-\frac{const}{t}\big)$, \\ elementary\end{tabular} \\ 

\hline
\textbf{\begin{tabular}[c]{@{}l@{}}Stochastic volatility \\ (diffusion model)\end{tabular}}
& \begin{tabular}[c]{@{}l@{}}$O(\sqrt{t})$, \\ \cite{CaWu03,MuNu11}\end{tabular} 
& \begin{tabular}[c]{@{}l@{}}$O(\sqrt{t})$, \\ \cite{MuNu11}\end{tabular}                                                               & \multicolumn{1}{!{\vrule width 2pt}c!{\vrule width 2pt}}{$\exp\big(-\frac{const}{t^{1-2\beta}}\big)$}
& \begin{tabular}[c]{@{}l@{}}$\exp\big(-\frac{const}{t}\big)$, \\ \cite{BeBuFl04}\end{tabular} \\ 

\hline
\Cline{2pt}{4-4}

\textbf{\begin{tabular}[c]{@{}l@{}}Jump diffusion/ \\ general semi-\\ martingale with\\ diff. component\end{tabular}}
& \begin{tabular}[c]{@{}l@{}}$O(\sqrt{t})$, \\ \cite{MuNu11}\end{tabular} 
& \begin{tabular}[c]{@{}l@{}}$O(\sqrt{t})$, \\ \cite{MuNu11}\end{tabular}
& \begin{tabular}[c]{@{}l@{}}$O(t)$ in L\'evy models,\\ \cite{MiTa12}\end{tabular}                                                                                              & \begin{tabular}[c]{@{}l@{}}$O(t)$, \\ \cite{BeCo12,CaWu03}\end{tabular}                                                            \\ 
\hline
\end{tabular}
\end{table}

\vspace{0.4cm}
Throughout the paper, we reserve the term out-of-the-money (OTM) for \emph{fixed}
OTM log-strike~$k>0$, to distinguish this regime from the {\it moderately} out-of-the-money regime that we now introduce. Our basic observation is that for 
\begin{equation}\label{eq:korg}
k \approx (const)t^\beta 
\end{equation}
the cases of $\beta  >\tfrac12$, resp.\ $\beta = 0$, are covered by the afore-discussed AATM, resp.\ OTM, results. This leaves
open a significant gap, namely $\beta \in (0,\tfrac{1}{2})$, which we call {\bf moderately out-of-the-money} (short: {\bf MOTM}). We have a threefold interest in this MOTM  regime,
\begin{equation}\label{eq:k}
    k \approx (const)t^\beta \quad   \text{ for }\ \beta \in ( 0,\tfrac{1}{2}).
\end{equation}
\, \, \, (i) First, it is very much related to the {\it reality of quoted (short-dated) option prices}, where strikes of option price data with acceptable bid-ask spreads
tend to accumulate ``around the money", as illustrated in Figure \ref{fig1}. It is then very natural to analyse strikes $k = O(t^\beta)$ for some $\beta > 0$, and there is no reason why quoted strikes should always be almost-ATM, which effectively means an extreme concentration around the money thanks to $\beta > 1/2$.

\begin{figure}[ht]
	\centering
  \includegraphics[width=0.9\linewidth]{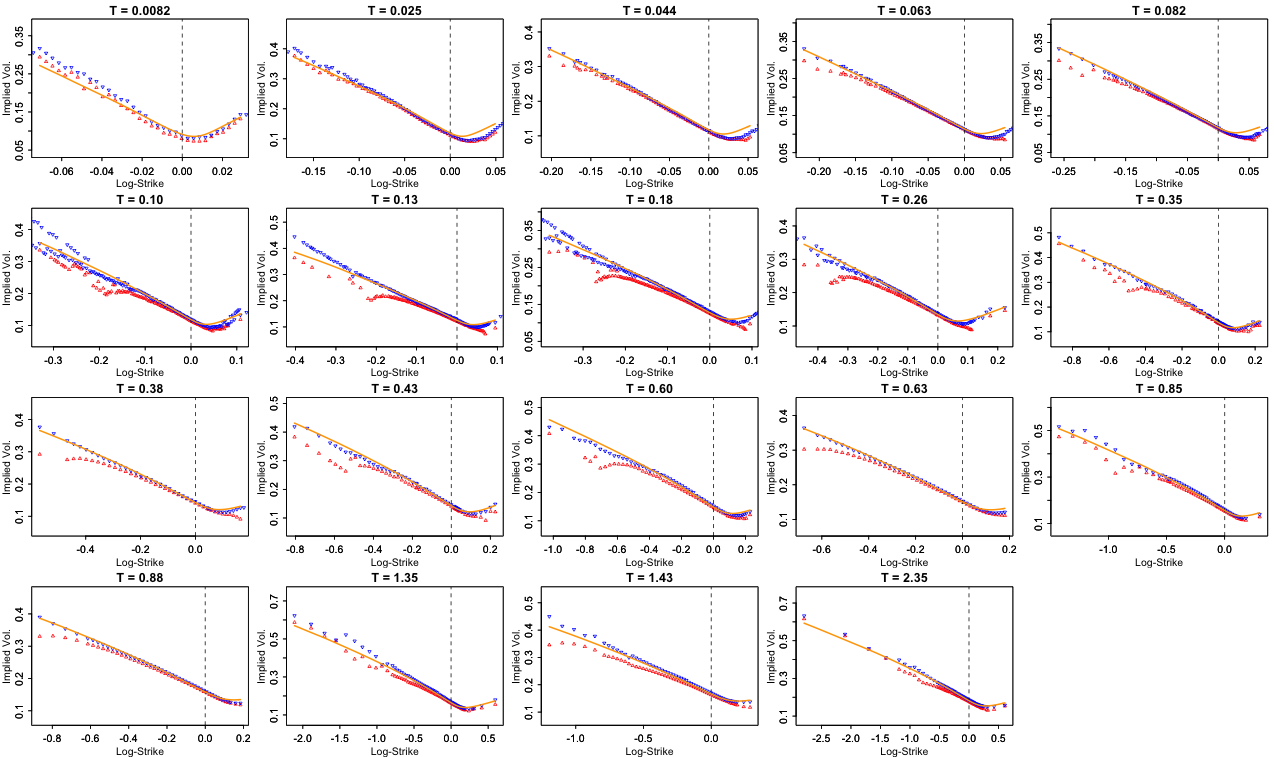} 
	\caption{SPX volatility smiles as of 14-Aug-2013 (courtesy of J. Gatheral). Strikes of options with small remaining time to maturity ($T=0.0082$) are about $e^{0.02}-1 \approx 2 \%$ around the money (spot), good data for a later 
	time $T=0.26$ already has already a range of $\approx 30 \%$,  the highest maturity $T=2.35$ has a range of $\approx 65 \%$ around the money. }
	\label{fig1}
\end{figure}

(ii) The second reason is {\it mathematical convenience}. In contrast to the genuine OTM regime (large deviation regime)  in which the rate function $\Lambda(k)$ is notoriously
difficult to analyse -- often related to geodesic distance problems -- MOTM naturally comes with a quadratic rate function and, most remarkably, higher order expansions
are always explicitly computable in terms of the model parameters. The terminology {\it moderately}-OTM (MOTM) is in fact in reference to {\it moderate deviations theory}, which effectively
interpolates between the central limit and large deviations regimes. (This also identifies the AATM regime as central limit regime, where asymptotics are precisely those of the Black-Scholes model,
which in turn is the rescaled Gaussian (in log-coordinates) limit of a general semimartingale model with diffusive component.)

(iii) Our third point is that MOTM expansions naturally involve quantities very familiar to practitioners, notably spot (implied) volatility, implied volatility skew and so on. 


%
%
%
%
%

In the Black-Scholes model, it is easy to check that
we have the following MOTM asymptotics,
\begin{equation*}
c_{\,\mathrm{BS}}( k_{t} \equiv \theta t^\beta, \, t;\sigma ) = \exp \bigg( -\frac{\theta
^{2}}{2\sigma ^{2}t^{1-2\beta }} \big(1+o(1)\big)\bigg), \quad t\downarrow 0.
\end{equation*}%
Loosely speaking, our main results assert that 
such relations (even of higher order) are true in great
generality for diffusion models, all quantities are computable and then related to 
implied volatility expansions. 
We note in passing that, for L\'evy models,
the regime \eqref{eq:k} has been studied in~\cite{MiTa12}; then,
call prices decay algebraically rather than exponentially. For recent related
results on \emph{fractional} stochastic volatility models, see~\cite{FoZh16,GuJaRo14}.
Guillin's work~\cite{Gu03} on small-noise moderate deviations of diffusions should
also be mentioned here; however, there the dynamics depend on a ``fast'' random environment
(with motivation from physics, and no obvious financial interpretation), and also the non-degeneracy assumption (D) from~\cite{Gu03} is not satisfied in our context.
\\

The rest of the paper is organized as follows. Section~\ref{se:dens} contains our main results,
which translate asymptotics for the transition density of the underlying into MOTM call price asymptotics.
The corresponding proofs are presented in Section~\ref{se:proofs}.
Section~\ref{se:iv} gives the implied volatility expansion resulting from our call price approximations.
Section~\ref{se:examples} applies our main results to standard examples, namely
generic local volatility models (Subsection~\ref{se:loc}), generic stochastic volatility
models (Subsection~\ref{se:stoch vol}), and the Heston model (Subsection~\ref{se:heston}). 
(As usual, the square-root degeneracy of the Heston model makes it difficult to apply results for
general stochastic volatility models, so we verify the validity of our results -- if formally applied to 
Heston -- by a direct ``affine'' analysis.) Finally, in Section~\ref{se:ge} we present a second approach at MOTM estimates, 
which employs the G\"artner-Ellis theorem from large deviation theory. Throughout we take zero rates, natural in view of our 
short-time consideration. Also, w.l.o.g.\ we normalize spot to $S_0=1$.

\section{MOTM option prices via density asymptotics}\label{se:dens}

We consider a {\it general stochastic volatility} model, i.e.\ a positive martingale $(S_t)_{t \ge 0}$ with dynamics
$$  
      dS_t = S_t \sigma (t, \omega)    dW_t
$$
started (w.l.o.g.) at $S_0 = 1$. We assume that the stochastic volatility (process) $\sigma$ itself is an It\^o-diffusion,
started at some deterministic value $\sigma (0, \omega) \equiv \sigma_0$, called {\it spot volatility}.
Recall that in any such stochastic volatility model, the local (or effective) volatility is 
defined by
$$
      \sigma^2_{\mathrm{loc}} (t,K) := \mathbb{E}[ \sigma^2 (t,\omega) | S_t = K ]. 
$$      
As is well-known, the equivalent local volatility model 
$$ 
      d\tilde S_t = \tilde S_t  \sigma^2_{\mathrm{loc}} (t,\tilde S_t) dW
$$      
has the property that $\tilde S_t = S_t$ (in law) for all fixed times. See~\cite{BrSh13} for some precise technical conditions under which this holds true.\footnote{The situation is very different with jumps, see~\cite{FrGeYo14}.}
In particular then, European option prices $C = C(K,t)$ match in both models. Recall also Dupire's formula in this context:
\begin{equation}\label{eq:dup}
\sigma _{\mathrm{loc}}^{2}( K,t) =\frac{\partial _{t}C( K,t)}{\frac{1}{2}%
K^{2}\partial _{KK}C( K,t)}, \quad t>0, K>0.
\end{equation}
We now state our two crucial conditions.

\begin{assumption}\label{ass:dens}
   For all $t>0$, $S_t$ has a continuous pdf $K\mapsto q(K,t)$, which behaves asymptotically as follows for small time:
\begin{equation}\label{eq:dens}
q( K,t) \sim e^{-\frac{\Lambda (k) }{t}}{t^{-1/2}}%
\gamma(k), \quad t\downarrow 0,
\end{equation}
uniformly for~$K=e^k$ in some neighbourhood of~$S_{0}=1$.
The energy function~$\Lambda$ is smooth in some neighbourhood of zero, with
$\Lambda(0)=\Lambda'(0)=0$.
Moreover, $\lim_{k\to 0}\gamma(k)=\gamma(0)>0$.
\end{assumption}

\begin{assumption}\label{ass:loc}
  For $t\downarrow0$ and $K\to S_0=1$, the local volatility function of~$(S_t)_{t\geq0}$
  converges to spot volatility:
  \[
    \lim_{\substack{K\to S_0 \\ t\downarrow0}} \sigma _{\mathrm{loc}}( K,t)  = \sigma_0.
  \]
\end{assumption}
The latter assumption is (in diffusion models) fairly harmless. The first assumption is potentially (very) difficult to check, but fortunately we can
rely on substantial recent progress in this direction \cite{DeFrJaVi14a, DeFrJaVi14b, Os15}.  We shall see in Section~\ref{se:stoch vol}  in detail 
that both assumptions indeed hold in generic stochastic volatility models. Let us also note the fundamental relation between spot-volatility $\sigma_0$ (actually equal
to implied spot vol $\sigma_{\mathrm{imp}} (0,0)$ here) and the Hessian of the energy function $\Lambda = \Lambda(k)$,
$$
   \sigma_0   =\Lambda''(0)^{-1/2}.
$$
(This is well-known (e.g. \cite{Du04}) and also follows from Proposition \ref{prop:osa} below.)
Now we state our main result.
We slightly generalize the log-strikes considered in~\eqref{eq:k},
replacing the constant 
by an arbitrary slowly varying function~$\ell$.

\begin{theorem}\label{thm:main}
Under Assumptions~\ref{ass:dens} and~\ref{ass:loc},
consider a moderately out-of-the-money call, in the sense that log-strike is
\begin{equation}\label{eq:k ell}
  k_{t}=t^\beta \ell(t), \quad t>0,
\end{equation}
where $\ell>0$ varies slowly at zero and $\beta\in(0,\tfrac12)$.

(i) The call price satisfies the moderate deviation estimate
\begin{align}
c( k_{t},t) &= \exp \bigg( -\frac{\Lambda ^{\prime
\prime }( 0) }{2}\frac{k_{t}^{2}}{t} \big(1+o(1)\big) \bigg) \notag \\
&= \exp \bigg( -\frac{1 }{2\sigma_0^2}\frac{k_{t}^{2}}{t} \big(1+o(1)\big)\bigg), \quad
t\downarrow 0. \label{eq:i}
\end{align}

(ii) If we restrict $\beta$ to $(0,\tfrac13)$, then the following
moderate second order expansion holds true:  

\begin{align}  
c( k_{t},t) &= \exp\bigg( - \frac{1}{2}\Lambda ^{\prime \prime }(
0) \frac{k_{t}^{2}}{t}- \frac{1}{6}\Lambda ^{\prime \prime \prime
}( 0) \frac{k_{t}^{3}}{t}\big(1+o( 1) \big)\bigg) \notag \\
&= \exp \bigg( -\frac{1 }{2\sigma_0^2}\frac{k_{t}^{2}}{t} \big(1- \frac{\mathcal{S}}{\sigma_0^2} k_t (1+o(1)) \big)\bigg), \quad
t\downarrow 0, \label{eq:md2}
\end{align}
with spot-variance $\sigma_0^2$, equal to $\sigma^2_{\mathrm{imp}}(0,0)$, and implied variance skew $\mathcal{S} = \frac{\partial }{\partial k}\big|_{k=0} \sigma_{\mathrm{imp}}^{2}(
  k,0).$

%
%
\end{theorem}
In particular, for $\ell\equiv\theta>0$, we have the (first order) expansion
\[
  t^{1-2\beta}\log c(\theta t^\beta,t) \sim -\frac{\theta^2}{2\sigma_0^2}, \quad
   t\downarrow 0,
\]
exhibiting a quadratic rate function $\theta\mapsto  \theta^2/2\sigma_0^2,$
typical of {\it moderate deviation problems.}\footnote{Recall that the MD rate function for a centered i.i.d.\ sequence $(X_n)_{n\geq1}$
is given by $\theta \mapsto \theta^2/(2 \mathrm{Var}(X_1))$. This is the
``moderate'' version of Cram\'er's theorem, see Theorem~3.7.1 in~\cite{DeZe98}.}\\


The quantities $ \Lambda ^{\prime \prime }( 0) ,\Lambda ^{\prime
\prime \prime }( 0),\dots $ appearing above are {\it always computable} from the initial values and the diffusion coefficients of the stochastic volatility model. This is in stark contrast to the OTM regime, where one needs $\Lambda(k)$, which is in general not available in closed form (with some famous exceptions, like the SABR model).
We quote the following result on $N$-factor models from  Osajima~\cite{Os15}, and refer to Section~\ref{se:stoch vol}
for detailed calculations in a two-factor stochastic volatility model.

\begin{proposition}\label{prop:osa}
 Assume that $(\log S, \sigma^1, \dots, \sigma^{N-1})$ is Markov, started at $(0, \bar\sigma_0)$ with $\bar\sigma_0\in\mathbb{R}^{N-1}$ and $\bar\sigma_0^1>0$, with stochastic volatility $\sigma \equiv \sigma^1$, where the generator has (non-degenerate) principal part $\sum a^{ij} \partial_{ij}$ in the sense that $a^{-1}$ defines a Riemannian metric. Then

\[
  \Lambda(k) = \frac1{2b_1}k^2 - \frac{b_2}{3b_1^3}k^3 + \bigg(-\frac{b_3}{4b_1^4} + \frac{b_2^2}{2b_1^5}\bigg)k^4 + O(k^5),\quad k\to 0,
\]
where the coefficients are given by
\begin{align*}
  b_1 &= \int_0^1 a^{11}(t,\bar\sigma_0)dt \\
  b_2 &= \frac32\int_0^1 (Va^{11})(t,\bar\sigma_0)dt \\
  b_3 &= 2\int_0^1 (V^2a^{11})(t,\bar\sigma_0)dt + \frac12\int_0^1\Gamma(a^{11}, a^{11})(t,\bar\sigma_0)dt,
\end{align*}
using the functions
\begin{align*}
  (Vf)(t,x) &= \sum_{i=1}^N a^{1i}(t,x)\int_t^1 \frac{\partial f}{\partial x_i}(s, x)ds, \\
  \Gamma(f,g)(t,x) &= \sum_{i,j=1}^N a^{ij}(t,x)\bigg(\int_t^1\frac{\partial f}{\partial x_i}(s, x)ds\bigg)\bigg(\int_t^1\frac{\partial g}{\partial x_j}(s, x)ds\bigg).
\end{align*}
\end{proposition} 

\begin{proof}
  See Osajima~\cite{Os15}, Theorem 1(1), with $T=1$.
\end{proof}
\medskip

The following result presents a higher-order expansion in the MOTM regime. It yields
an asymptotically equivalent expression for call prices (and not just logarithmic asymptotics).

\begin{theorem}\label{thm:higher}
  Under the assumptions of Theorem~\ref{thm:main}, we set $k_t=t^\beta \ell(t)$
  with $\beta\in(0,\tfrac12)$ and $\ell>0$ slowly varying at zero, as above.
  Then the logarithm of the
   call price has the following refined MOTM expansion:
  \begin{multline}\label{eq:refined}
    \log c(k_t,t) = - \sum_{m=2}^{\lfloor 1/\beta \rfloor} \frac{\Lambda^{(m)}(0)}{m!}
      \frac{k_t^m}{t} \\
      + \Big(2\beta-\frac32\Big) \log\frac1t
      - 2 \log \ell(t) + \log \big(\gamma(0)v_0^2\big) + o(1), \quad t\downarrow0.
  \end{multline}
  Equivalently,
  \[
    c(k_t,t) \sim \gamma(0) v_0^2 \frac{t^{3/2-2\beta}}{\ell(t)^2}
      \exp\bigg(- \sum_{m=2}^{\lfloor 1/\beta \rfloor} 
      \frac{\Lambda^{(m)}(0)}{m!}\frac{k_t^m}{t} \bigg), \quad t\downarrow0.
  \]
\end{theorem}

\medskip
  
  If $1/\beta$ is not an integer, then
  $k_t^m/t$ tends to infinity for $m=\lfloor 1/\beta \rfloor$, of order
  $t^{\beta \lfloor 1/\beta \rfloor-1}$ (up to a slowly varying factor).
  If $1/\beta$ is an integer, on the other hand, then the last summand of the sum
  $\sum_{m=2}^{\lfloor 1/\beta \rfloor}$
  in~\eqref{eq:refined} is of order $\ell(t)$, which means that the following term
  $\log(1/t)$ may be asymptotically larger. The upper summation limit $\lfloor 1/\beta \rfloor$
  thus ensures that no irrelevant (i.e., $o(1)$) terms are contained in the sum.
  Note that we have $\lfloor 1/\beta \rfloor=2$ for $\beta\in(\tfrac13,\tfrac12)$,
and $\lfloor 1/\beta \rfloor\geq 3$ for $\beta\in(0,\tfrac13)$, and
so~\eqref{eq:refined} is consistent with~\eqref{eq:i}
resp.~\eqref{eq:md2}.
\\

The passage from the (derivatives of the) energy
to ATM derivatives of the implied volatility in the short time limit is best conducted via the 
Berestycki-Busca-Florent (short: BBF) formula that was proved in~\cite{BeBuFl02}. (That said, 
theses relations are also a direct consequence of our expansions, as is pointed out in Section~\ref{se:iv}.)
In this regard, we have


\begin{theorem}\label{thm:skew}
  Suppose that~$\Lambda$ is a function with the properties required in Assumption~\ref{ass:dens},
  with $\Lambda''(0)=\sigma_0^{-2}=v_0^{-1}$,
  and that the Berestycki-Busca-Florent formula
$\sigma_{\mathrm{imp}}^{2}( 0,k)  =k^{2}/2\Lambda(k)$ holds. 
  Then the small-time ATM implied variance skew and curvature, respectively, relate to~$\Lambda$ via
  \begin{equation}\label{eq:skew L}
  \mathcal{S}:=\frac{\partial }{\partial k}\Big|_{k=0} \sigma_{\mathrm{imp}}^{2}(
  k,0)  =-\frac{1}{3}\frac{\Lambda^{\prime \prime \prime}(0)}{\Lambda^{\prime \prime}(0)^2}
  \end{equation}
  and
    \begin{equation}\label{eq:curvature L}
  \mathcal{C}:=\frac{\partial^2}{\partial k^2}\Big|_{k=0} \sigma_{\mathrm{imp}}^{2}(k,0) = \frac{\tfrac{2}{3}\Lambda^{\prime\prime\prime}(0)^2-\tfrac{1}{2}\Lambda^{\prime\prime\prime\prime}(0)\Lambda^{\prime\prime}(0)}{3\Lambda^{\prime\prime}(0)^3}.
  \end{equation}
\end{theorem}

\begin{proof}
  By the BBF formula and our smoothness assumptions on~$\Lambda$,
  \begin{align*}
  \sigma_{\mathrm{imp}}^{2}(k,0)  &=\frac{k^{2}}{2\Lambda (
  k) } =k^{2}\bigg(\Lambda ^{\prime \prime }( 0) k^{2}+\frac{%
  1}{3}\Lambda ^{\prime \prime \prime }( 0) k^{3} +\frac{1}{12}\Lambda ^{\prime \prime \prime \prime }( 0) k^{4}+O(
  k^{5})\bigg)^{-1} \\
  & = \frac{1}{\Lambda^{\prime\prime}(0)}
      -\frac{1}{3}\frac{\Lambda^{\prime\prime\prime}(0)}{\Lambda^{\prime\prime}(0)^2}k 
      +\bigg(\frac{\tfrac{1}{9}\Lambda^{\prime\prime\prime}(0)^2-\tfrac{1}{12}\Lambda^{\prime\prime\prime\prime}(0)\Lambda^{\prime\prime}(0)}{\Lambda^{\prime\prime}(0)^3}\bigg)k^2, \quad k\to 0.
  \end{align*}
  This implies~\eqref{eq:skew L} and~\eqref{eq:curvature L}.
\end{proof}

\begin{remark} Proposition \ref{prop:osa} combined with Theorem \ref{thm:skew} allows to compute skew and curvature (and higher derivatives of the implied volatility smile, if desired)
directly from the coefficients of a general stochastic volatility model. Related formulae for  ``general'' (even non-Markovian) models also appear in the work of Durrleman (Theorem 3.1.1.\ in~\cite{Du04}; see also~\cite{Du10}).
While not written in the setting of general Markovian diffusion models, and hence not in terms of the energy function~$\Lambda$, they inevitably give the same results if applied
to given parametric stochastic volatility models (see Section~3.1 in~\cite{Du04}). However, this work comes with some (seemingly) uncheckable assumptions, the drawback of which is discussed in 
in Section~2.6 of~\cite{Du04}. 
\end{remark}
%
%
%

\section{Proofs of the main results}\label{se:proofs}

\begin{proof}[Proof of Theorem~\ref{thm:main}]
Since the density of~$S_t$ satisfies $q=\partial _{KK}C$, we have,
by Dupire's formula~\eqref{eq:dup},
\begin{equation*}
C( K,t) =\int_{0}^{t}\partial _{s}C( K,s)
ds=\int_{0}^{t}\frac{1}{2}q( K,s) K^{2}\sigma _{\mathrm{loc}}^{2}(
K,s) ds.
\end{equation*}%
Then, for $K_t=e^{k_t}$ with $k_t \downarrow 0$ as stated, we apply
Assumption~\ref{ass:loc} as follows:
\begin{align*}
C( K_t,t)  &=\int_{0}^{t}\frac{1}{2}q( K_t,s) K_t^{2}\sigma
_{\mathrm{loc}}^{2}( K_t,s) ds \\
&\sim \frac{\sigma_0^2}{2} \int_{0}^{t}q( K_t,s) ds, \quad t\downarrow 0.
\end{align*}
%
And then, using local uniformity of our density expansion~\eqref{eq:dens},
\begin{align}
C(K_t,t) &\sim \frac{\sigma_0^2\gamma(0)}{2} \int_{0}^{t}e^{-\frac{\Lambda
( k_{t}) }{s}}{s^{-1/2}}dt \notag \\
&=\frac{\sigma_0^2\gamma(0)}{2} t\int_{0}^{1}e^{-\frac{\Lambda ( k_{t}) }{xt}}\left( xt\right) 
{^{-1/2}}dx \notag \\
&=\frac{\sigma_0^2\gamma(0)}{2} t^{1/2}\int_{0}^{1}e^{-\frac{\Lambda ( k_{t}) }{xt}}x^{-1/2}%
dx. \label{eq:int for laplace}
\end{align}
Since $\Lambda$ is smooth at zero, and using the fact that $\Lambda ( 0) =\Lambda ^{\prime
}( 0) =0$, we have
\[
  \frac{\Lambda(k_t)}{t} \sim \frac12 \Lambda''(0) \frac{k_t^2}{t}\to\infty
  \quad \text{as}\ t\downarrow 0.
\]
For small $t$, the integrand in~\eqref{eq:int for laplace} is  thus concentrated near $x=1$,
and by the Laplace method (Theorem~7.1 in~\cite{Ol74})
\begin{equation*}
\int_{0}^{1}e^{-\frac{\Lambda ( k_{t}) }{xt}}x{^{-1/2}}dx \sim
\frac{t}{\Lambda ( k_{t}) } \exp\Big(-\frac{\Lambda ( k_{t}) }{t} \Big).
\end{equation*}%
Therefore,
\begin{align}
  C(K_t,t) &\sim \frac{\sigma_0^2\gamma(0)}{2}
  \frac{t^{3/2}}{\Lambda ( k_{t}) } \exp\Big(-\frac{\Lambda ( k_{t}) }{t} \Big) \notag \\
  &\sim v_0^2 \gamma(0) \frac{t^{3/2}}{\ k_{t}^2 } \exp\Big(-\frac{\Lambda ( k_{t}) }{t} \Big), 
  \quad  t\downarrow 0, \label{eq:C as}
\end{align}
which implies (recall the notation~$c$ resp.~$C$ from~\eqref{eq:not})
\begin{align}
  -\log c(k_t,t) &= \frac{\Lambda ( k_{t}) }{t} - \log \frac{t^{3/2}}{ k_{t}^2 }
    +O(1) \label{eq:c} \\
    &= \frac1t \left( \frac12 \Lambda''(0)k_t^2 + \frac16 \Lambda'''(0)k_t^3
     + O(k_t^4)\right) + O\Big(\log \frac{k_t^2}{t^{3/2}}\Big). \notag
\end{align}
to prove (i) and (ii), we thus need to argue that $k_t^2/t$
 dominates $\log(k_t^2 t^{-3/2})$ if $\beta \in(0,\tfrac12)$, and that $k_t^3/t$ dominates
 $\log(k_t^2 t^{-3/2})$ if $\beta \in(0,\tfrac13)$.
For $m \in \{2,3\}$, we calculate
\begin{align*}
  \frac{k_t^m/t}{|\log(k_t^2 t^{-3/2})|} &=
    \frac{t^{m\beta-1}\ell(t)^m}{|\log(t^{2\beta-3/2}\ell(t)^2)|} \\
   &= \frac{t^{m\beta-1}\ell(t)^m}{|(2\beta-\tfrac32)\log t + 2\log \ell(t)|}.
\end{align*}
{}From Proposition 1.3.6 (i) in~\cite{BiGoTe87} we know that $\log \ell(t) = o(\log t)$, and so
\[
   \frac{k_t^m/t}{|\log(k_t^2 t^{-3/2})|} \sim
     \frac{t^{m\beta-1}\ell(t)^m}{|(2\beta-\tfrac32)\log t|}, \quad t\downarrow0.
\]
This tends to infinity for $m=2$ and $\beta\in(0,\tfrac12)$, and 
for $m=3$ and $\beta\in(0,\tfrac13)$, as desired.
\end{proof}
\medskip

Inspecting the preceding proof, it is easy to see that we can expand 
$\log c(k_t,t)$ further:
\medskip

\begin{proof}[Proof of Theorem \ref{thm:higher}] 
  Taking logs in~\eqref{eq:C as} yields
  \[
    \log c(k_t,t) = -\frac{\Lambda(k_t)}{t} + \log \frac{t^{3/2}}{k_t^2}
      + \log\big(\gamma(0) v_0^2\big) + o(1).
  \]
  Then~\eqref{eq:refined} follows by Taylor expanding~$\Lambda$. Note that $k_t^m/t = o(1)$
  for $m \geq \lfloor 1/\beta \rfloor +1$.
\end{proof}

\section{Implied volatility}\label{se:iv}

We now investigate some relations of MOTM call price expansions with implied volatility. 

\begin{corollary}\label{cor:iv}
   Under the assumptions of Theorem~\ref{thm:main},
   put $k_t=t^\beta \ell(t)$  with $\beta\in (0,\tfrac13)$ and $\ell>0$ slowly varying.
   Then the implied volatility   has the following MOTM expansion:      
  \begin{align}
     \sigma_{\mathrm{imp}}(k_t,t)  &= \sigma_0 - \tfrac16  \sigma_0^3\Lambda'''(0)
       k_t(1+o(1)) \label{eq:iv1} 
  \end{align}
\end{corollary}
\begin{proof}
   We use our main result (Theorem~\ref{thm:main})
   in conjunction with a transfer result of Gao and Lee~\cite{GaLe14}.
   As the call price tends to zero, we are in case~``$(-)$'' of~\cite{GaLe14}
   (defined on p.~354 of that paper).
    The notation $L$, $V$ of~\cite{GaLe14} means  $L=-\log c(k_t,t)$ 
    resp.\ $V=t^{1/2} \sigma_{\mathrm{imp}}(k_t,t)$, the dimensionless
   implied volatility. Then  Corollary~7.2 of~\cite{GaLe14} implies that
   \begin{equation}\label{eq:V}
      V = \frac{k_t}{\sqrt{2L}}\big(1+O(t^{1-2\beta-\varepsilon})\big)
        + O(t^{5/2-4\beta-\varepsilon}), \quad t\downarrow0.
   \end{equation}
   Here, $\varepsilon>0$ denotes an arbitrarily small constant that serves
   to eat up slowly varying functions in $O$-estimates.
   By part~(ii) of Theorem~\ref{thm:main}, we have
   \[
      2L =  \frac{1}{\sigma_0^2}\frac{k_t^2}{t} + \frac{\Lambda'''(0)}{3}
         \frac{k_t^3}{t}(1+o(1)).
   \]
   Inserting this into~\eqref{eq:V} gives
   \begin{equation*}
     \sigma_{\mathrm{imp}}(k_t,t)= t^{-1/2}k_t\left(\sigma_0 \frac{t^{1/2}}{k_t}
        -\frac{\sigma_0^3 \Lambda'''(0)}{6} t^{1/2}\big(1+o(1)\big) \right) + O(t^{2-4\beta-\varepsilon}),
   \end{equation*}
   which yields~\eqref{eq:iv1}.
\end{proof}
\medskip

The above corollary has some interesting consequences. Under the sheer assumption that implied volatility 
has a first order Taylor expansion for small maturity and small log-strike of the form
\begin{equation}\label{eq:taylor}
  \sigma_{\mathrm{imp}}(k,t) = \sigma_0 + \partial_k \sigma_{\mathrm{imp}}(0,0)\, k
   + o(k) +O(t), \quad t\downarrow0,\ k=o(1),
\end{equation}
then of course in the MOTM regime, we have $ t \ll k_t$, and so the $k$-term dominates the $O(t)$-term, which
in turn identifies the implied variance skew as
\begin{equation}
        \mathcal{S} = \lim_{t\downarrow0} \frac{2\sigma_0}{k_t}(\sigma_{\mathrm{imp}}(k_t,t)-\sigma_0).
         \label{eq:iv S}  
\end{equation}
On the other hand, Corollary \ref{cor:iv} now implies that the right-hand side of (\ref{eq:iv S}) equals
$-\tfrac{1}{3} \sigma_0^4 \Lambda^{\prime \prime \prime}(0)$. 
We have thus arrived at an alternative proof of the skew representation~\eqref{eq:skew L} in terms of the energy function,
without using the BBF formula. (The curvature and higher order derivatives of the ATM smile can be dealt with similarly, if desired.)

\section{Examples}\label{se:examples}

\subsection{Generic local volatility models}\label{se:loc}

Clearly, Assumption~\ref{ass:loc} is satisfied for any local volatility model,
assuming continuity of the local volatility function. We now discuss
Assumption~\ref{ass:dens}, and show how to compute our MOTM expansions.
First consider the time-homogeneous local volatility model
\begin{equation}\label{eq:loc dyn}
  dS_t = \sigma(S_t)S_t dW_t, \quad S_0 = 1,
\end{equation}
where $\sigma$ is $C^2$ on $(0,\infty)$. An expansion
of the pdf~$q(\cdot,t)$ of~$S_t$ has been worked out in Gatheral et al.~\cite{GaHsLaOuWa12}.
They assume growth conditions on $\sigma$ and its derivatives,
which can be alleviated by the principle of not feeling the boundary
(Appendix~A of~\cite{GaHsLaOuWa12}).
 Proposition~2.1 of~\cite{GaHsLaOuWa12}
says that
\[
  q(e^k,t) \sim \frac{e^{-k}u_0(0,k)}{\sqrt{2\pi t}}
  \exp\bigg(-\frac{\Lambda(k)}{t}\bigg)   ,\quad t\downarrow0,
\]
uniformly in~$k$,
where the energy function is given by (cf.\ Varadhan~\cite{Va67})
\[
  \Lambda(k) = \frac12 \bigg( \int_0^k \frac{dx}{\sigma(e^x)} \bigg)^2,
\]
and
\[
  u_0(0,k) = \sigma(1)^{1/2}\sigma(e^k)^{-3/2}e^{-k/2}.
\]
(Recall that we normalize spot to $S_0=1$ throughout.)
This shows that Assumption~\ref{ass:dens} is satisfied, with
$\gamma(0)=1/(\sqrt{2\pi}\sigma(1))$. To evaluate the expansions from
Theorem~\ref{thm:main}, we compute the derivatives of~$\Lambda$:
\[
  \Lambda'(k) = \frac{1}{\sigma(e^k)} \int_0^k \frac{dx}{\sigma(e^x)},
  \qquad
  \Lambda''(k) = \frac{1}{\sigma(e^k)^2} - \frac{e^k\sigma'(e^k)}{\sigma(e^k)^2}
     \int_0^k \frac{dx}{\sigma(e^x)},
\]
and
\[
  \Lambda'''(k) = -\frac{3e^k \sigma(e^k)}{\sigma(e^k)^3}
  +\left(\frac{2e^{2k}\sigma'(e^k)^2}{\sigma(e^k)^3}-\frac{\sigma(e^k)''}{\sigma(e^k)^2} \right)\int_0^k \frac{dx}{\sigma(e^x)},
\]
which yield
\[
  \Lambda''(0) = \frac{1}{\sigma(1)^2} =\frac{1}{\sigma(S_0)^2}
\]
and
\begin{equation}\label{eq:loc 3}
  \Lambda'''(0) = -\frac{3\sigma'(1)}{\sigma(1)^3} =-\frac{3\sigma'(S_0)}{\sigma(S_0)^3}.
\end{equation}
Alternatively, these expressions can be obtained from Proposition~\ref{prop:osa}.
Since the assumptions of Theorem~\ref{thm:main} are satisfied, we obtain
the following MOTM call price estimates, where $k_{t}=\theta t^{\beta }$ and $\theta>0$:
\begin{align*}
  c( k_{t},t) &= \exp \bigg( -\frac{\theta^2}{2 \sigma(1)^2 t^{1-2\beta }}%
  \big(1+o(1)\big)\bigg), \quad \beta\in(0,\tfrac12),\ t\downarrow0, \\
  c( k_{t},t)  
   &= \exp \bigg( -\frac{\theta ^{2}}{2\sigma(1)^2 t^{1-2\beta }}-\frac{\sigma'(1)}{2\sigma(1)^3}\frac{\theta ^{3}}{t^{1-3\beta }}\big(1+o(1)\big)\bigg), \quad
    \beta\in(0,\tfrac13),\  t\downarrow0 .
\end{align*}
Recall that we denote by~$\mathcal{S}$ the (limiting small-time ATM) implied \emph{variance} skew,
and so the implied \emph{volatility} skew is given by
$\mathcal{S}/2\sigma_0$, which equals $\mathcal{S}/2\sigma(1)=\mathcal{S}/2\sigma(S_0)$ in the model~\eqref{eq:loc dyn}.
From~\eqref{eq:skew L} and~\eqref{eq:loc 3}
we recover that the \emph{local} skew $\sigma'(1)=\sigma'(S_0)$ equals
twice the implied volatility skew,
\[
  \sigma'(S_0) = 2 \times \frac{\mathcal{S}}{2\sigma(S_0)},
\]
as observed by Henry-Labord\`ere (Remark~5.2 in~\cite{He09}). Generic time-inhomogeneous local volatility models
\[
  dS_t = \sigma(S_t,t)S_t dW_t
\]
could be treated very similarly, using the heat kernel expansion
in Section~3 of~\cite{GaHsLaOuWa12}, itself taken from Yosida~\cite{Yo53}.

\subsection{Generic stochastic volatility models}\label{se:stoch vol}

We now discuss the results of Section~\ref{se:dens} in generic stochastic volatility models.
Rigorous conditions under which stochastic volatility models satisfy
Assumption~\ref{ass:dens} can be found in~\cite{DeFrJaVi14a,Os15}.
The function~$\Lambda$ is given by the Riemannian metric associated to the model:
$2\Lambda ( k) $ is the squared geodesic distance from $%
( S_{0}=1,\sigma _{0}) $ to $\left\{ ( K,\sigma )
:\sigma >0\right\} $ with $K=e^{k}$.
Theorem~2.2 in Berestycki, Busca, and Florent~\cite{BeBuFl04} gives conditions under which
Assumption~\ref{ass:loc}, concerning convergence of local volatility, is true.

Now we describe how the expressions appearing in the expansions from Theorem~\ref{thm:main}
can be computed explicitly in a generic two-factor stochastic volatility model
\begin{align}
  dS_t &= S_t\sqrt{V_t} dW_t, \quad S_0 = 1, \nonumber \\
  dV_t &= (\dots) dt + \eta\sqrt{V_t}\nu(V_t) dZ_t, \quad V_0 = v_0 > 0,   \label{SVdynamics}
\end{align}
where $\nu\colon\mathbb{R}\to\mathbb{R}$ and $d\langle W,Z\rangle_{t}=\rho\, dt$.
The Heston model ($\nu(v)\equiv 1$) and the 3/2-model ($\nu(v)=v$; see~\cite{Le00})
are special cases.
The infinitesimal generator $L$ of the stochastic process $(S,V)$, neglecting first order terms, can be written as
\[
  Lf \approx \frac12\operatorname{Tr}\left(\left(\begin{matrix} g_{11} & g_{12} \\ g_{21} & g_{22}\end{matrix}\right) D^2 f\right),\quad f\in C^2(\mathbb{R}^2),
\]
where $D^2 f$ denotes the Hessian matrix of $f$, and the coefficient matrix $g=(g_{ij})$ in this model is given by
\[
  g = \left(\begin{matrix} v & \rho\eta v\nu(v) \\ \rho\eta v\nu(v) & \eta^2 v\nu(v)^2\end{matrix}\right).
\]
We define the constants $b_1 = g_{11}|_{v=v_0} = v_0$ and $b_2 = \frac34\sum_{i=1}^2 g_{1i} \partial_i g_{11}|_{v=v_0} = \frac34\rho\eta v_0\nu(v_0)$.
If we assume that the coefficients in (\ref{SVdynamics}) are nice enough to justify application of the (marginal) density expansion obtained in \cite{DeFrJaVi14a} or part~(2) of Theorem 1 in~\cite{Os15}, 
we get the desired small-time density expansion \eqref{eq:dens}. Moreover, thanks to Proposition~\ref{prop:osa},
\[
  \Lambda(k) = \frac1{2b_1}k^2 - \frac{b_2}{3b_1^3}k^3 + O(k^4)
\]
in a neighborhood of 0. Therefore the quantities $\Lambda''(0)= v_0^{-1} = \sigma_0^{-2}$ and $\Lambda'''(0) = -\frac32\rho\eta\nu(v_0)/v_0^2$ can easily be computed, as well as the small-time ATM implied variance skew 
\[
  \mathcal{S} = -\frac{v_0^2}{3}\Lambda'''(0) = \frac{\rho\eta}{2}\nu(v_0).
\]
Thus, all quantities appearing in our expansions (Theorem~\ref{thm:main},
Corollary~\ref{cor:iv}) have very simple expressions in terms 
of the model parameters.

\subsection{The Heston model}\label{se:heston}

This section contains an application of the results of Sections~\ref{se:dens} and~\ref{se:iv}
to the familiar case of the Heston model,
where many explicit ``affine'' computations are possible.
The Heston model is not in the scope of the general results implying Assumptions~\ref{ass:dens}
and~\ref{ass:loc}, which we recalled at the beginning of Section~\ref{se:stoch vol}.
We will explain how both assumptions can be verified rigorously by a dedicated analysis; full
details would involve rather dull repetition of arguments that are found in the literature
in a very similar form, and are therefore omitted.
The model dynamics are
\begin{align*}
dS_{t}& =S_{t}\sqrt{V_{t}}dW_{t},\qquad S_{0}=1,  \notag \\
dV_{t}& =-\kappa \left( V-\bar{v}\right) dt+\eta\sqrt{V_{t}}dZ_{t},\qquad V_{0}=v_{0}>0,
\label{E:H}
\end{align*}
where $\bar{v}, \kappa, \eta>0$, and $d\langle W,Z\rangle
_{t}=\rho\, dt$ with $\rho \in (-1,1)$.
According to Forde and Jacquier~\cite{FoJa09}, the first order OTM (large deviations) behavior
of the call prices is
\begin{equation}\label{eq:FJ}
t\log c_{\mathrm{He}}( k,t) \sim -\Lambda_{\mathrm{He}}( k), \quad k>0\ \text{fixed},\ t\downarrow0,
\end{equation}%
where $\Lambda_{\mathrm{He}}$ is the (not explicitly available) Legendre
transform of 
\begin{align}
\Gamma( p)  &=\frac{v_{0}p}{\eta ( \bar{\rho}\cot \left(
\eta \bar{\rho}p/2\right) -\rho )}  \notag \\
&= \frac{v_{0}p}{\eta \big( \bar{\rho}\big(\frac{1}{\eta \bar{\rho}p/2}%
+O(p)\big)-\rho \big)}  \notag \\
&=\frac{v_{0}p}{\frac{1}{p/2}-\eta \rho +O(p)}  \notag \\
&=\frac{v_{0}p^{2}/2}{1-p\eta \rho /2+O( p^{2}) } \notag \\
&=\frac{v_{0}p^{2}}{2}\left( 1+p\eta \rho /2+O( p^{2}) \right). \label{eq:Gamma expans}
\end{align}
(We use the standard notation $\bar{\rho}^2=1-\rho^2$.)
This expansion implies
\begin{equation}\label{eq:Heston''}
\Gamma^{\prime \prime }( 0) =v_{0}=\sigma_0^2.
\end{equation}
The locally uniform density asymptotics~\eqref{eq:dens} hold,
as seen from an easy modification of the arguments in Forde, Jacquier,
and Lee~\cite{FoJaLe12}.
There, the Fourier representation of the call price was analysed by the saddle point method
to obtain a refinement of~\eqref{eq:FJ}. Proceeding completely analogously
for the Fourier representation of the pdf of~$S_t$, we get the density approximation
\begin{align*}
  q_{\mathrm{He}}(e^k,t) &= \frac{e^{-k}}{2\pi t}\int_{-\infty-i p^*(k)}^{\infty-i p^*(k)}
    \mathrm{Re}\big(e^{i k u/t} \phi_t(-u/t) \big) du \\
    &= \exp\bigg( -\frac{\Lambda_{\mathrm{He}}(k)}{t} \bigg) 
     \frac{U(p^*(k))}{\sqrt{2\pi \Gamma''(k)}}t^{-1/2}\big(1+o(1)\big), \quad t\downarrow0,
\end{align*}
locally uniformly in~$k$,
where $\phi_t$ is the characteristic function of~$X_t=\log S_t$,
and~$p^*$ and~$U$ are defined on p.~693 of~\cite{FoJaLe12}.
(Note that~\cite{FoJaLe12} uses the notation~$\Lambda,\Lambda^*$ instead
of our $\Gamma,\Lambda_{\mathrm{He}}$.)
From~\eqref{eq:Heston''} and the fact that $U(p^*(0))=U(0)=1$, we see that
the factor~$\gamma(k)$ from~\eqref{eq:dens} converges to
\begin{equation}\label{eq:c0}
  \gamma(0) = \frac{1}{\sqrt{2\pi}\sigma_0}
\end{equation}
as $k\to0$.

To verify Assumption~\ref{ass:loc} (convergence of local volatility),
the Dupire formula~\eqref{eq:dup} can be subjected to an analysis similar
to~\cite{DeFrGe13,FrGe15}. More precisely, $\partial_{KK}C(K,t)$ in the numerator of~\eqref{eq:dup}
is the pdf of~$S_t$, the analysis of which we have just described.
Virtually the same saddle point approach can be applied to the 
numerator $\partial_t C(K,t)$, yielding convergence of the quotient to~$\sigma_0^2$.

We now calculate our MOTM asymptotic expansions for the Heston model.
The Legendre transform~$\Lambda_{\mathrm{He}}$ is given by
$\Lambda_{\mathrm{He}}( k) =\sup_{x}\{kx-\Gamma( x) \}$ with
maximizer $x^{\ast }=x^*( k) $. From general facts on Legendre
transforms,%
\begin{equation*}
  \Lambda_{\mathrm{He}}^{\prime \prime }( k) =\frac{1}{\Gamma^{\prime \prime
  }( x^*( k) ) }.
\end{equation*}
Since $x^*( 0)=0 $, we have
\begin{equation*}
  \Lambda_{\mathrm{He}}^{\prime \prime }( 0) =\frac{1}{\Gamma^{\prime \prime
  }( 0) }=\frac{1}{v_{0}}.
\end{equation*}
{}From Theorem~\ref{thm:main}, with $k_{t}=\theta t^{\beta }$ and $\theta>0$, we
then obtain the MOTM call price estimate
\begin{equation}\label{eq:Heston c}
  c_{\mathrm{He}}( k_{t},t) = \exp \bigg( -\frac{\theta^2}{2v_{0}t^{1-2\beta }}%
  \big(1+o(1)\big)\bigg), \quad t\downarrow0.
\end{equation}
As for the second order expansion, from the expansion~\eqref{eq:Gamma expans} of $\Gamma$ we clearly see that
\begin{equation*}
  \Gamma^{\prime \prime \prime }( 0) =\frac{3}{2}v_{0}\eta \rho .
\end{equation*}%
On the other hand, a general Legendre computation gives
\begin{equation*}
\Lambda_{\mathrm{He}}^{\prime \prime \prime }( k) =-\left( \frac{1}{\Gamma
^{\prime \prime }( x^*( k) ) }\right) ^{2}\Gamma
^{\prime \prime \prime }( x^*( k) )\, (x^*)^{\prime }(
k) =-(\Lambda_{\mathrm{He}}^{\prime \prime }\left( k\right) )^{3}\, \Gamma^{\prime
\prime \prime }( x^*( k) ).
\end{equation*}%
Therefore,
\begin{equation*}
\Lambda_{\mathrm{He}}^{\prime \prime \prime }( 0) =-\frac{3}{2}\frac{\eta \rho 
}{v_{0}^{2}},
\end{equation*}
in accordance with the expression
for generic two-factor models, found in Section~\ref{se:stoch vol}.
For $\beta \in(0,\tfrac13)$, Theorem~\ref{thm:main} (ii) thus implies the second order expansion 
\begin{equation}\label{eq:md2 H}
c_{\mathrm{He}}( k_{t},t)  
= \exp \bigg( -\frac{\theta ^{2}}{2v_{0}t^{1-2\beta }}+\frac{\eta
\rho }{4v_{0}^{2}}\frac{\theta ^{3}}{t^{1-3\beta }}\big(1+o(1)\big)\bigg), \quad t\downarrow0 .
\end{equation}
By Theorem~\ref{thm:higher} and~\eqref{eq:c0}, we obtain the following refined call price
expansions:
\begin{align}
  \log c_{\mathrm{He}}(k_t,t) &= -\frac{1}{2\sigma_0^2} \frac{k_t^2}{t} + \Big(\frac32-2\beta\Big)\log t
  + \log \frac{\sigma_0^3}{\sqrt{2\pi}} + o(1),
  \quad \beta\in(\tfrac13,\tfrac12), \label{eq:H higher} \\
  \log c_{\mathrm{He}}(k_t,t) &= -\frac{1}{2\sigma_0^2} \frac{k_t^2}{t}
    + \frac{\eta\rho}{4v_0^2} \frac{k_t^3}{t}
   + \Big(\frac32-2\beta\Big)\log t
  + \log \frac{\sigma_0^3}{\sqrt{2\pi}} + o(1),
   \quad \beta\in(\tfrac14,\tfrac13) \label{eq:H higher 2}.
\end{align}
{}From the relation~\eqref{eq:skew L} between implied variance skew and $\Lambda
^{\prime \prime \prime }(0)$, we obtain the explicit expression
$\mathcal{S}_{\mathrm{He}}=\eta \rho /2$ for the skew. This agrees
with Gatheral~\cite{Ga06}, p.~35.
The implied volatility expansion~\eqref{eq:iv S}
becomes
\begin{equation}\label{eq:H iv}
  \sigma_{\mathrm{imp}}(k_t,t)=\sigma_0 + \frac{\eta\rho}{4\sigma_0}k_t\big(1+o(1)\big) , \quad t\downarrow0.
\end{equation}
Figure~\ref{fig:iv} shows a good fit of this approximation, even for maturities
that are not very small.

\begin{figure}[h]
  \centering
  \includegraphics{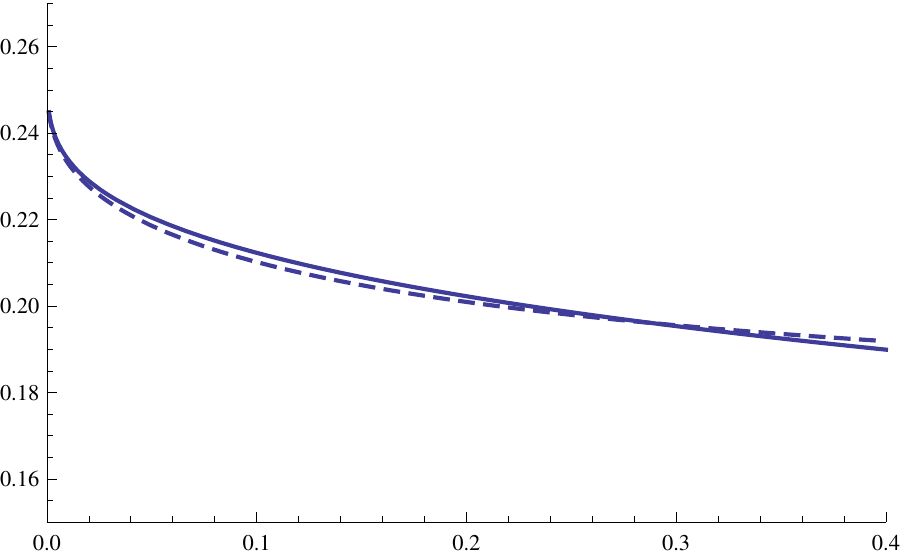}
  \caption{Illustration of our implied volatility expansion for the Heston model, with
  $\ell\equiv\theta =0.4$
  and $\beta=0.3$. Thus, log-strike
  equals $k=0.4\, t^{0.3}$.
   The model parameters are
  $\bar{v}=0.0707,$ $\kappa=0.6067,$ $\eta=0.2928,$ $\rho=-0.7571,$ $v_0=0.0654$ (i.e.,
  $\sigma_0=0.2557$), and $S_0=1$.
  The horizontal axis is time.
  The dashed curve is the exact MOTM implied volatility $\sigma_{\mathrm{imp}}(k_t,t)$.
  The solid curve is the approximation $\sigma_0 + \frac{\eta\rho}{4\sigma_0}k_t$
  on the right hand side of~\eqref{eq:H iv}.}
  \label{fig:iv}
\end{figure}


\section{MOTM option prices via the G\"artner-Ellis theorem}\label{se:ge}

In this section we discuss a different approach at small-time moderate deviations.
While yielding only first order results, its conditions are usually easy to check
for models with explicit characteristic function.
Assumptions~\ref{ass:dens} and~\ref{ass:loc} are not in force here.

 Recall that,
in the classical setting of sequences of i.i.d.\ random variables, a moderate deviation
analogue of Cram\'er's theorem can be deduced by applying the G\"artner-Ellis theorem
to an appropriately rescaled sequence (see~\cite{DeZe98}, Section~3.7).
The MD short time behavior of diffusions can be subjected to a similar analysis.
Consider the log-price $X_t=\log S_t$ with $X_0=0$
and mgf (moment generating function)
\begin{equation}\label{eq:mgf}
  M(p,t) := \mathbb{E}[e^{p X_t}].
\end{equation}
\begin{assumption}\label{ass:ge}
    For all $\beta\in(0,\tfrac12)$, the rescaled mgf satisfies
  \begin{equation}\label{eq:ass md}
    \lim_{t\downarrow0}  t^{1-2\beta} \log M(t^{\beta-1}p,t) = \tfrac12\sigma_0^2 p^2,
    \quad p\in\mathbb{R}.
  \end{equation}
\end{assumption}
We expect that this assumption holds for diffusion models in considerable generality.
It is easy to check that~\eqref{eq:ass md} holds for the Heston model, either by
its explicit characteristic function, or, more elegantly, from the associated Riccati equations;
see the forthcoming PhD thesis of A.~Pinter for details.
Thus, the results of the present section provide an alternative proof of the first order MOTM behavior~\eqref{eq:Heston c}
of Heston call prices.

Heuristically, Assumption~\ref{ass:ge} can be derived from the density asymptotics in Assumption~\ref{ass:dens}, which in turn hold in quite general diffusion settings~\cite{DeFrJaVi14a,DeFrJaVi14b}:
\begin{align}
  M(t^{\beta-1}p,t) &= \int e^{t^{\beta-1}p x}q(x,t)dx \notag \\
  &\approx \int \exp\Big( t^{\beta-1}p x - \frac{\Lambda(x)}{t} \Big) dx \label{eq:M1} \\
  &\approx \int \exp\Big( t^{\beta-1}p x - \frac{\Lambda''(0)x^2}{2t} \Big) dx  \label{eq:M2} \\
  &=\exp\Big(  t^{\beta-1}p x - \frac{x^2}{2\sigma_0^2t}
      \Big|_{x= \sigma_0^2 p t^\beta} (1+o(1)) \Big) \label{eq:M Lap} \\
  &= \exp \Big( \tfrac12 \sigma_0^2 p^2  t^{2\beta-1}\big(1+o(1)\big) \Big),
    \quad t\downarrow0. \label{eq:M appr}
\end{align}
In~\eqref{eq:M1}, we ignored that the density expansion~\eqref{eq:dens} might not be valid globally in space;
this might be made rigorous by estimating $q(x,t)$ by a Freidlin-Wentzell LD argument
for~$x$ sufficiently large. As for~\eqref{eq:M2}, we can expect concentration near $x\approx0$,
since $\Lambda(x)$ increases with~$|x|$. Finally,
\eqref{eq:M Lap}, and thus~\eqref{eq:M appr}, follows from a (rigorous) application of the Laplace method.
If~\eqref{eq:M appr} is correct, then~\eqref{eq:ass md} clearly follows.

The critical moment of~$S_t$ is defined by
\[
   p_+(t) := \sup\{ p \geq 0: M( p,t)<\infty\}.
\]
It is obvious that
\begin{equation}\label{eq:nec}
  \lim_{t\downarrow0} t^{1-\beta} p_+(t)=\infty 
\end{equation}
is necessary for~\eqref{eq:ass md}, i.e.,
$p_+(t)$ must grow faster than $t^{\beta-1}$ as $t\downarrow0$.
In the Heston model, e.g., the critical moment is of order $p_+(t)
\approx 1/t \gg t^{\beta-1}$ for small~$t$, as follows from inverting~(6.2) in~\cite{KR11}.
On the other hand, we do not expect our results to be of much use in the presence of jumps. Indeed, suppose
that~\eqref{eq:mgf} is the mgf of an exponential L\'evy model. Then $p_+(t)\equiv p_+$
does not depend on~$t$, and is finite for most models used in practice. Therefore,
\eqref{eq:nec} cannot hold, and so Assumption~\ref{ass:ge} is not satisfied.
 The Merton jump diffusion model is one of the
few L\'evy models of interest that have $p_+=\infty$, but it is easy to
check that it does not satisfy~\eqref{eq:ass md}, either.

After this discussion of Assumption~\ref{ass:ge},
we now give an asymptotic estimate for the distribution function of~$X_t$ (put differently, MOTM
\emph{digital call} prices) in Theorem~\ref{thm:ge}. Then we translate this result to MOTM \emph{call}
prices in Theorem~\ref{thm:transfer}.
If desired, higher order terms in~\eqref{eq:ass md} will give refined asymptotics
in Theorem~\ref{thm:ge}, using Gulisashvili and
Teichmann's recent refinement of the G\"artner-Ellis theorem~\cite{GuTe15}.
 It might not be trivial to translate
the resulting expansions into call price asymptotics, though.
For other asymptotic results on option prices using the G\"artner-Ellis theorem,
see, e.g, \cite{FoJa09,FoJa11}.
\begin{theorem}\label{thm:ge}
  Under Assumption~\ref{ass:ge} (and without any further assumptions on our model),
  for $k_t = \theta t^{\beta}$ with $\beta\in(0,\tfrac12)$ and $\theta>0$, we have a first order MD estimate
  for the cdf of~$X_t$:
  \begin{equation}\label{eq:md1}
    \mathbb{P}[X_t \geq k_t] = \exp\left(-\frac{1}{2\sigma_0^2} \frac{k_t^2}{t}
    \big(1+o(1)\big)\right), \quad t\downarrow0.
  \end{equation}
\end{theorem}
\begin{proof}
  Define
  \[
    Z_t := t^{-\beta} X_t, \quad \text{with mgf} \quad M_Z(s,t) = \mathbb{E}[e^{sZ_t}],
  \]
  and
  \[
    a_t := t^{1-2\beta} = o(1), \quad t\downarrow 0.
  \]
  Then~\eqref{eq:ass md} is equivalent to
  \[
    \Gamma_Z(p) := \lim_{t\downarrow0} a_t \log M_Z(p/a_t,t) 
       = \tfrac12 \sigma_0^2 p^2,
    \quad p\in\mathbb{R}.
  \]
  Since $\Gamma_Z$ is finite on~$\mathbb{R}$, the G\"artner-Ellis theorem (Theorem~2.3.6
  in~\cite{DeZe98}) implies that
  $(Z_t)_{t\geq0}$ satisfies an LDP (large deviation principle)
  as $t\downarrow0$, with rate~$a_t$ and good rate
  function~$\Lambda_Z$, the Legendre transform of $\Gamma_Z$. Trivially, $\Lambda_Z$ is quadratic, too:
  \begin{align*}
    \Lambda_Z(x) &= \sup_{p \in\mathbb{R}}(p x - \Gamma_Z(p)) \\
      &= \sup_{p \in\mathbb{R}}(p x - \tfrac12 \sigma_0^2 p^2)
      = \frac{x^2}{2\sigma_0^2}, \quad x\in\mathbb{R}.
  \end{align*}
  Now fix $\theta>0$. Applying the lower estimate of the LDP to $(\theta,\infty)$ yields
  \begin{align*}
    \liminf_{t\downarrow0}a_t \log \mathbb{P}[Z_t \geq \theta] &\geq
    \liminf_{t\downarrow0}a_t \log \mathbb{P}[Z_t > \theta] \\
     &\geq -\Lambda_Z(\theta) =  -\frac{\theta^2}{2\sigma_0^2},
  \end{align*}
  and applying the upper estimate to $[\theta,\infty)$ yields
  \[
     \limsup_{t\downarrow0}a_t \log \mathbb{P}[Z_t \geq \theta] \leq -\frac{\theta^2}{2\sigma_0^2},
  \]
  and so
  \[
    \lim_{t\downarrow0}a_t \log \mathbb{P}[Z_t \geq \theta] = -\frac{\theta^2}{2\sigma_0^2}.
  \]
  This is the same as~\eqref{eq:md1}.
\end{proof}
\medskip

As in the LD/OTM regime, first order cdf asymptotics translate readily into call price asymptotics.
The proof of the following result is similar to~\cite{Ph10}, p.~30f (concerning
the LD regime) and~\cite{CaCo14}, Theorem~1.5.
In the MD/MOTM regime, one can replace
the condition~(1.19) of~\cite{CaCo14} by a mild condition on the moments of the model.

\begin{theorem}\label{thm:transfer}
  Let $S=e^X$ be a continuous positive martingale. Assume that,  for all $p\geq1$, its $p$-th moment
  explodes at a positive time (infinity included).
  By this we mean that there is a positive $t^*(p)$ such that the mgf
  $\,\mathbb{E}[\exp(pX_t)]$ is finite for all $t\in[0,t^*(p)]$. Let $v_0=\sigma_0^2>0$.
  Then the following are equivalent:
  \begin{itemize}
  \item[(i)] For $k_t=\ell(t) t^{\beta}$, with $\beta\in(0,\tfrac12)$
  and $\ell>0$ slowly varying at zero, it holds that
  \[
     \mathbb{P}[X_t \geq k_t] = \exp\left(-\frac{1}{2v_0} \frac{k_t^2}{t}
    \big(1+o(1)\big)\right), \quad t\downarrow0.
  \]
  \item[(ii)] Under the assumptions of~(i), we have
  \begin{equation}\label{eq:transfer call}
    c(k_t,t) = \exp\left(-\frac{1}{2v_0}  \frac{k_t^2}{t}
    \big(1+o(1)\big)\right), \quad t\downarrow0.
  \end{equation}
  \end{itemize}
\end{theorem}
\begin{proof}
  First assume~(i).
  Let $\varepsilon>0$ and define $\tilde{k}_t=(1+\varepsilon)k_t$. Then
  \begin{align}
    c(k_t,t) &\geq \mathbb{E}[(e^{X_t}-e^{k_t})^+\,
    \mathbf{1}_{\{X_t \geq \tilde{k}_t\}}] \notag \\
    &\geq (e^{\tilde{k}_t}-e^{k_t})^+ \mathbb{P}[X_t \geq \tilde{k}_t]. \label{eq:proof1}
  \end{align}
  The first factor is
  \[
    (e^{\tilde{k}_t}-e^{k_t})^+ = (\tilde{k}_t-k_t + O(k_t^2))^+ = \varepsilon k_t +O(k_t^2).
  \]
  For the second factor in~\eqref{eq:proof1}, we apply~(i) with $\tilde{k}_t$:
  \[
    \lim_{t\downarrow0} \frac{t}{\tilde{k}_t^2} \log \mathbb{P}[X_t \geq \tilde{k}_t]
      = - \frac{1}{2v_0}.
  \]
  Therefore,
  \begin{align*}
    \liminf_{t\downarrow0} \frac{t}{k_t^2} \log c(k_t,t) &\geq \lim_{t\downarrow0}
      \frac{t}{k_t^2}\Big(- \frac{1}{2v_0} \frac{\tilde{k}_t^2}{t} \big(1+o(1)\big)  \Big) \\
    &= -\frac{(1+\varepsilon)^2}{2v_0}.
  \end{align*}
  Now let $\varepsilon\downarrow0$ to get the desired lower bound for $c(k_t,t)$.
  
  As for the upper bound, we let $p>1$ and note that, by definition of~$p\mapsto t^*(p)$,
  we have 
  $\mathbb{E}[S_{t}^{p+1}]<\infty$ for all $t\in[0,t^*(p+1)]$. Define
  $\bar{S}_t=\sup_{0\leq u\leq t} S_u$ for $t\geq0$. By Doob's inequality
  (Theorem~3.8 in~\cite{KaSh91}), we have
  \[
    \mathbb{P}[\bar{S}_{t^*(p+1)} \geq s] \leq \frac{\mathbb{E}[S_{t^*(p+1)}^{p+1}]}{s^{p+1}}, \quad s>0.
  \]
  Hence $\bar{S}_{t^*(p+1)}$ has a finite $p\kern .08em$th moment:
  \[
     \mathbb{E}[(\bar{S}_{t^*(p+1)})^p] = p \int_0^\infty s^{p-1} 
     \mathbb{P}[\bar{S}_{t^*(p+1)} \geq s] ds <\infty.
  \]
  By the dominated convergence theorem and the continuity of~$S$, we thus conclude
  \begin{equation}\label{eq:proof2}
    \lim_{t\downarrow0} \mathbb{E}[S_t^p] = S_0^p.
  \end{equation}
  Now let $1/p+1/q=1$ and apply H\"older's inequality:
  \begin{align*}
    c(k_t,t) &= \mathbb{E}[(e^{X_t}-e^{k_t})^+\, \mathbf{1}_{\{X_t\geq k_t\}}] \\
    &\leq \mathbb{E}[((e^{X_t}-e^{k_t})^{+})^p]^{1/p}\ \mathbb{P}[X_t\geq k_t]^{1/q} \\
    &\leq \mathbb{E}[S_t^p]^{1/p}\ \mathbb{P}[X_t\geq k_t]^{1/q}.
  \end{align*}
   By~\eqref{eq:proof2} and (i), we obtain
   \[
     \limsup_{t\downarrow0} \frac{t}{k_t^2} \log c(k_t,t) \leq
     \frac1q \limsup_{t\downarrow0} \frac{t}{k_t^2} \log \mathbb{P}[X_t\geq k_t]
     = -\frac{1}{2 q v_0}.
   \]
   Now let $p\uparrow\infty$, i.e., $q\downarrow1$.
   The same argument yields the lower bound of the implication
   (ii) $\Longrightarrow$ (i). The remaining upper bound of (ii) $\Longrightarrow$ (i) is shown very similarly
   to the lower bound of the implication
   (i) $\Longrightarrow$ (ii).
\end{proof}
\medskip

\bibliographystyle{siam}
\bibliography{gerhold}

\end{document}